\documentclass[runningheads]{llncs} 

\usepackage{amsmath,amssymb,mathrsfs,bbold,stmaryrd}

\usepackage{hyperref}
\hypersetup{
    colorlinks=true,
    linkcolor=blue,
    filecolor=magenta,      
    urlcolor=cyan,
}

\newenvironment{propproof}[1][\proofname]{\proof[#1]}{\endproof}

\usepackage{tikz}
\usetikzlibrary{decorations.pathreplacing}
\usetikzlibrary{shapes.misc}
\usepackage{pifont}
\tikzset{cross/.style={cross out, draw=black, minimum size=2*(#1-\pgflinewidth), inner sep=0pt, outer sep=0pt},cross/.default={2pt}}
\usepackage{comment}
\usepackage{physics}
\usepackage{graphicx}

\usepackage{stmaryrd}
\usepackage{xcolor}
\usepackage{subcaption}

\usepackage{algorithm}

\usepackage{multirow}

\usepackage{cleveref}

\usepackage{cancel}

\newcommand{\Z}{{\mathbb{Z}}}
\newcommand{\Q}{{\mathbb{Q}}}
\newcommand{\R}{{\mathbb{R}}}
\newcommand{\C}{{\mathbb{C}}}

\newcommand{\mS}{{{S}}}

\newcommand{\Lc}{{{\rm Lc}}}

\usepackage[normalem]{ ulem }
\usepackage{soul}

\begin{document}

\title{${L}^{\infty}$-norm computation for linear time-invariant systems depending on parameters}

\author{
Alban Quadrat \inst{1} \and
Fabrice Rouillier\inst{1} \and 
Grace Younes\inst{2}
}

\authorrunning{A. Quadrat et al.}

\institute{
{Sorbonne Universit\'e and Universit\'e de Paris, CNRS, IMJ-PRG, Inria Paris, F-75005 Paris, France\\ \email{\{alban.quadrat, fabrice.rouillier\}@inria.fr}}\and
{Department of Science and Engineering, Sorbonne University Abu Dhabi, Abu Dhabi, UAE\\ \email{grace.younes@sorbonne.ae}}} 
\maketitle     

\begin{abstract}

This paper focuses on representing the $L^{\infty}$-norm of finite-dimensional linear time-invariant systems with parameter-dependent coefficients. Previous studies tackled the problem in a non-parametric scenario by simplifying it to finding the maximum $y$-projection of real solutions $(x, y)$ of a system of the form $\Sigma=\{P=0, \,  \partial P/\partial x=0\}$, where $P \in \Z[x, y]$. To solve this problem, standard computer algebra methods were employed and analyzed \cite{bouzidi2021computation}.

In this paper, we extend our approach to address the parametric case. We aim to represent the ``maximal" $y$-projection of real solutions of $\Sigma$ as a function of the given parameters.
To accomplish this, we utilize cylindrical algebraic decomposition. This method allows us to determine the desired value as a function of the parameters within specific regions of parameter space.
\end{abstract}
\section{Introduction}
A major issue in control theory aims to design a system, called a controller, which stabilizes a given unstable system, called a plant, and then optimizes the performances of the closed-loop system obtained by adding the controller to the plant in a feedback loop. Within the frequency domain approach to linear time-invariant systems, in a series of seminal papers, Zames shows that the standard concept of stability corresponds to a system defined by a transfer matrix (which defines the dynamics between the inputs and the outputs of the system) whose entries belong to the Hardy algebra $H^\infty$ formed by the holomorphic functions in the right half complex plane which are bounded for the supremum norm. By the maximum modulus principle, it means that the $H^\infty$-norm corresponds to the $L^\infty$-norm of the restriction of the transfer matrix to the imaginary axis $i \, \R$. In the case of finite-dimensional systems, i.e., systems defined by linear ordinary differential equations, the coefficients of the corresponding transfer matrix are rational functions of a complex variable $s$. Zames' approach was the starting point for the development of the so-called robust control theory or $H^\infty$-control theory (see \cite{Zhou} and the references therein), a major achievement in control theory, which is nowadays commonly used by companies working in the area of control theory. In this approach, a property of a system is said to be robust if it is valid for all systems defined in a ball around the original system with a small radius for the $H^\infty$-norm. This concept of robustness is crucial in control theory since it can be used to mathematically model the different uncertainties and errors naturally occurring in any mathematical model of a physical system. 

The computation of the $H^\infty$-norm of finite-dimensional linear time-invariant systems, i.e., the $L^\infty$-norm of the restriction of a matrix with complex univariate rational function entries to the purely imaginary axis $i \, \R$, is thus a fundamental problem in robust control theory. Unfortunately and contrary to the $L^2$-norm (see, e.g., \cite{Zhou}), it is not possible to simply characterize this $L^\infty$-norm using simple closed-form expressions depending on the system coefficients. 

Previous studies (see \cite{bouzidi2021computation} and the references therein) investigated this problem by reformulating it as the search for the maximum $y$-projection of real solutions $(x,y)$ for a system of two polynomial equations $\Sigma = \left\{ {P =0, \partial P/\partial x = 0}\right\}$, where $P \in \mathbb Z[x,y]$. To solve this problem in a certified manner, standard computer algebra methods (such as RUR, Sturm-Habich sequences, and certified root isolation of univariate polynomials) were used, and three different algorithms were proposed with their complexity analysis \cite{bouzidi2021computation}.

In this paper, we investigate the $H_{\infty}$-norm of linear time-invariant systems featuring coefficients dependent on parameters. To accomplish this, we expand upon the methodology introduced in \cite{bouzidi2021computation}, adapting it to handle parametric scenarios where $P$ belongs to $\mathbb{Z}\left[\alpha_1, \ldots, \alpha_d \right][x, y]$. Here, $\alpha_1, \ldots, \alpha_d$ represent the parameters under consideration. In this setting, the “maximal” $y$-projection of real solutions of $\Sigma$ can now be viewed as a semi-algebraic function of the given parameters. We show that it can be computed as part of a “cylindrical algebraic decomposition” of the parameter's space adapted to polynomials computed with an ad-hoc projection strategy.

In the following sections, we will begin by outlining the problem at hand. Subsequently, we will provide a concise overview of the research conducted in \cite{bouzidi2021computation}. This will serve as a foundational context before delving into the extension of this work to scenarios involving parameters within the polynomial coefficients. It is important to note that addressing these parameter-dependent cases presents challenges that cannot be straightforwardly addressed through numerical methods.
\section{Problem description}\label{Probdesc}

We first introduce a few standard notations and definitions. If $\mathbb{k}$ is a field and $P \in \mathbb{k}[x,y]$, then $Lc_{var}(P)$ is the \emph{leading coefficient} of $P$ with respect to the variable $var \in \{x,y\}$ and $\deg_{var}(P)$ the \emph{degree} of $P$ in the variable $var \in \{x,y\}$. We also denote by $\deg(P)$ the \emph{total degree} of $P$. Moreover, let $\pi_x: \mathbb{R}^{2} \longrightarrow \mathbb{R}$ be the projection map from the real plane $\mathbb{R}^{2}$ onto the $x$-axis, i.e., $\pi_x(x,y) =x$ for all $(x, \, y) \in \mathbb{R}^{2}$. For $P, Q \in \mathbb{k}[x, y]$, let ${\rm gcd}(P, Q)$ be the greatest common divisor of $P$ and $Q$, $I:=\langle P, Q \rangle$ the ideal of $\mathbb{k}[x, y]$ generated by $P$ and $Q$, and $V_{\mathbb{K}}(I):=\{(x, \, y) \in \mathbb{K}^2 \; | \; \forall \; R \in I: \; R(x, y)=0\}$, where $\mathbb{K}$ is a field containing $\mathbb{k}$. Finally, let $\mathbb{C}_+:=\{s \in \mathbb{C} \; | \; \Re(s) > 0\}$ be the \emph{open right-half plane} of $\mathbb{C}$.

\begin{definition}[\cite{Doyle,Zhou}]
Let ${RH}_{\infty}$ be the $\mathbb{R}$-algebra of all the proper and stable rational functions with real coefficients, namely:
$${RH}_{\infty}:= \left \{\dfrac{n}{d} \; | \; n, d \in \mathbb{R}[s], \; {\rm gcd}(n, d)=1, \, \deg_s(n) \leq \deg_s(d), V_{\mathbb{C}}(\langle d \rangle) \cap \mathbb{C}_+=\emptyset \right \}.$$
\end{definition}

An element $g$ of ${RH}_{\infty}$ is holomorphic and bounded on $\mathbb{C}_+$, i.e.,
$$\parallel g \parallel_{\infty} \, := \, \sup_{s \in \mathbb{C}_+} |g(s)| < +\infty,$$
${RH}_{\infty}$ is a sub-algebra of the \emph{Hardy algebra} ${H}_{\infty}(\mathbb{C}_+)$ of bounded holomorphic functions on $\mathbb{C}_+$. The \emph{maximum modulus principle} of complex analysis yields:
$$\parallel g \parallel_{\infty} \, = \, \sup_{\omega \in \mathbb{R}} |g(i \, \omega)|.$$
Note that this equality shows that the function $g_{|i \, \mathbb{R}}: i \, \omega \in i \, \mathbb{R} \longmapsto g(i \, \omega)$ belongs to the \emph{Lebesgue space} ${L}_{\infty}(i \, \mathbb{R})$ or, more precisely, to the following $\mathbb{R}$-algebra
$$\begin{array}{ll}
{RL}_{\infty} := \left 
\{\dfrac{n(i \, \omega)}{d(i \, \omega)} \; | \; n, d \in \mathbb{R}[i \, \omega], \; {\rm gcd}(n, d)=1, \right.  & \deg_\omega(n) \leq \deg_\omega(d), \; V_{\mathbb{R}}(\langle d \rangle)=\emptyset
\},
\end{array}$$
i.e., the algebra of real rational functions on the imaginary axis $i \, \mathbb{R}$ which are proper and have no poles on $i \, \mathbb{R}$, or simply, the algebra of real rational functions with no poles on $i \, \mathbb{P}^1(\mathbb{R})$, where 
$\mathbb{P}^1(\mathbb{R}):=\mathbb{R} \cup {\infty}$.

We can extend the above ${L}_{\infty}$-norms defined on functions of ${RH}_{\infty}$ (resp., ${RL}_{\infty}$) to matrices as follows. Let  $G \in {RH}_{\infty}^{u \times v}$ (resp., $G \in {RL}_{\infty}^{u \times v}$, $\mathbb{R}(s)^{u \times v}$), i.e., $G$ is a $u \times v$ matrix with entries in ${RH}_{\infty}$ (resp., ${RL}_{\infty}$, $\mathbb{R}(s)$) and let $\Bar{\sigma}\left ( \cdot \right)$ denote the largest \emph{singular value of a complex matrix}.  Then, we can define:
$$\parallel G \parallel_{\infty} \, := \,  \sup_{s \in \mathbb{C}_+} \Bar{\sigma}\left(G(s) \right) \;  \left(\mbox{resp.,} \, \parallel G \parallel_{\infty} \, := \,  \sup_{\omega \in \mathbb{R}} \Bar{\sigma}\left(G(i \, \omega) \right) \right).$$ 
If $G \in {RH}_{\infty}^{u \times v}$, then, as above, we have $\parallel G \parallel_{\infty} \, = \, \sup_{\omega \in \mathbb{R}} \Bar{\sigma}\left(G(i \, \omega) \right)$. 

The \emph{conjugate} $\Tilde{G}$ of $G \in \mathbb{R}(s)^{u \times v}$ is defined by $\Tilde{G}(s):=G^{T}(-s)$. 

The next proposition gives a first characterization of $\parallel G \parallel_{\infty}$.

\begin{proposition}[\cite{kanno2006validated}]\label{prop:Kanno-Smith}
Let $\gamma > 0$, $G \in \mathbb{R}(s)^{u \times v}$ be such that 
$G_{|i \, \mathbb{R}} \in {RL}_{\infty}^{u \times v}$ and let us consider $\Phi_{\gamma}(s)= \gamma^2 \, I_v-\Tilde{G}(s) \, G(s)$. Then, $\gamma > \parallel G \parallel_{\infty}$ if and only if $\gamma
> \Bar{\sigma}\left(G(i \, \infty)\right)$ and $\det(\Phi_{\gamma}(i \, \omega)) \neq 0$ for all $\omega \in \mathbb{R}$.  
\end{proposition}

Let $n(\omega, \gamma)$ and $d(\omega)$ be two coprime polynomials over $\mathbb{R}[\omega,\gamma]$ satisfying: 
\begin{equation}\label{def:Phi}
\det(\Phi_{\gamma}(i \, \omega))=\dfrac{n(\omega, \gamma)}{d(\omega)}.    
\end{equation}
Note that $\det(\Phi_{\gamma}(s))$ is a real function in $s^2$ and $\gamma^2$, and thus, $\det(\Phi_{\gamma}(i \, \omega))$ is a real function in $\omega^2$ and $\gamma^2$.

Hence, to compute the maximal singular value of $G(i \, \omega)$, we have to compute the maximal real value $\gamma$ satisfying that a real value $\omega$ exists such that  $\det(\Phi_{\gamma}(i \, \omega))$ vanishes. Since $\det(\Phi_{\gamma}(i \, \omega))$ is a real rational function of $\omega$ and $\gamma$ (in fact a real rational function of $\omega^2$ and $\gamma^2$ by the parity of $\Phi_{\gamma}$), we can write $\det(\Phi_{\gamma}(i \, \omega))=n(\omega, \gamma)/d(\omega)$, where $n(\omega, \gamma) \in \R[\omega, \gamma]$ and $d(\omega) \in \R[\omega]$ are coprime. Since $G$ has no poles on the imaginary axis, 
$d(\omega)$ does not vanish on $\R$. Hence, to compute the $L^{\infty}$-norm of $F$, it suffices to compute the maximal real value $\gamma$ such that there exists at least one real value $\omega$ for which $n(\omega, \gamma)$  vanishes. {\em Hence, we are led to studying the $\gamma$-extremal points (and thus, the \emph{critical points}) of the following real plane algebraic curve:} 
\begin{equation}\label{eq:Curve_C}
{\mathcal C}:=\{(\omega, \, \gamma) \in \R^2 \; | \; n(\omega, \gamma)=0\}. 
\end{equation}

\begin{proposition}\label{prop:caracterization_infty_norm}
Let $G \in {RL}_{\infty}^{u \times v}$ and $n \in \mathbb{R}[\omega, \gamma]$ be defined by (\ref{def:Phi}). We denote by $\Bar{n} \in \mathbb{R}[\omega,\gamma]$ the square free part of $n$. Then, we have: 
$$\parallel G \parallel_{\infty} =\max \left \{ \pi_{\gamma} \left ( V_{\mathbb{R}} \left(
\left \langle 
\Bar{n}, \frac{\partial \Bar{n}}{\partial \omega} 
\right \rangle 
\right)\right) \cup V_{\mathbb{R}}\left(\langle Lc_{\omega}(\Bar{n}) \rangle \right)  \right \}.$$
\end{proposition}

\begin{corollary}\label{RLinftybounded}
Let $G$ $\in$ ${R}{L}_{\infty}^{u \times v}$ and $n \in \mathbb{R}[\omega, \gamma]$ be the numerator of $$\det(\gamma^2 \, I_v- \Tilde{G}(i \, \omega) \, G(i \, \omega))$$ defined by (\ref{def:Phi}). Then, the real $\gamma$-projection $\pi_{\gamma}(V_{\mathbb{R}}(\langle n \rangle))$ of $V_{\mathbb{R}}(\langle n \rangle)$ is bounded by $\parallel G\parallel_{\infty}$.
\end{corollary}

According to Proposition~\ref{prop:caracterization_infty_norm}, given $G \in {RL}_{\infty}^{u \times v}$, the problem of computing $\parallel G \parallel_{\infty}$ can be reduced to the computation of the maximal $\gamma$-projection of the real solutions of the following bivariate polynomial system: 
\begin{equation}\label{eq:Sigma}
\Sigma:=\left \{\Bar{n}(\omega,\gamma), \frac{\partial \Bar{n}(\omega, \gamma)}{\partial \omega} \right \}.    
\end{equation}

This amounts to finding the maximal positive real number $\gamma_{\star}$ such that the plane algebraic curve $n(\omega, \gamma)=0$ have points of the form $(\omega, \; \gamma_{\star})$ with $\omega$ being real numbers, i.e., such that $n(\omega,\gamma_{\star})=0$ has at least one real root $\omega$. This study was done in \cite{bouzidi2021computation} for $n \in \mathbb Z[\omega, \gamma]$. In the upcoming section, we will provide a concise summary of this study before delving into our study of the problem where $n \in \mathbb{Z}[\alpha][\omega, \gamma]$, with $\alpha=\alpha_1, \ldots, \alpha_d$ representing a set of real parameters.

\section{Proposed methods in the non-parametric case}\label{Problem description: non parametric case}
Without loss of generality, we will suppose that $n$ is squarefree in $\mathbb{Z}[\omega,\gamma]$. We also denote $x := \omega$, $y:=\gamma$ and $P:=n(\omega, \gamma)$.

Based on standard computer algebra methods and implementations, \cite{bouzidi2021computation} developed methods for the study of the computation of the $L^{\infty}$-norm of finite-dimensional linear time-invariant systems. This problem is reduced to the computation of the maximal $y$-projection of the real solutions $(x, y)$ of the zero-dimensional system of two bivariate polynomial equations defined by:
{ \begin{equation}\label{eq:Sigma}
 \Sigma = V_{\mathbb R}\left(\left\langle P,\; Q\right \rangle\right), \quad P\, \in \mathbb{Z}[x,y], \quad Q =\frac{\partial P}{\partial x}.
 \end{equation} 
}
Given coprime polynomials $P$ and $Q$ of degrees bounded by $d$ and coefficient bitsize bounded by $\tau$, \cite{bouzidi2021computation} proposes two approaches for the computation of the maximal $y$-projection of the real solutions of $\Sigma$. The first one uses a \emph{linear separating form} and the second the \emph{real root counting} of a univariate polynomial with algebraic coefficients. We briefly explain the functionality of both approaches before moving to the parametric case.

\emph{Separating linear form}. This approach is a direct bivariate solving with polynomials with rational coefficients. In this context, two methods were developed requiring putting (\ref{eq:Sigma}) in \emph{a local generic position},  i.e., finding a separating linear form $y + a \, x$ that defines a shear of the coordinate system $(x, \, y)$, i.e., $(x, \, y) \longmapsto (x, \, t- a \, x)$, so that no two distinct solutions of the sheared system defined by $\Sigma_a=\big\{P(x,  t- a \, x) = 0,\, Q(x,  t- a \, x) = 0 \big\}$ are horizontally aligned. This approach has long been used in the computer algebra literature. 
As shown in \cite{bouzidi2015separating,bouzidi2016solving}, a separating linear form $y + a \, x$, with $a\in \{0,\hdots,2 \, d^4\}$, can be computed. As studied in \cite{bouzidi2016solving}, we can then use a \emph{Rational Univariate Representation} (RUR) for the sheared system $\Sigma_a$ and compute  \emph{isolating boxes} for its real solutions. We can simply apply this approach (i.e., the so-called \emph{RUR method}) to the polynomial system associated with the $L^{\infty}$-norm computation problem and then choose the maximal $y$-projection of the real solutions of the system. This value is represented by its \emph{isolating interval} with respect to the univariate polynomial embodying the $y$-projection of the system solutions, i.e., the so-called \emph{resultant polynomial} $\Res(P, Q, x)$. The complexity analysis shows that this algorithm performs $\Tilde{{O}}_B\big(d_{y} \, d_{x}^3 \, (d_{y}^2+d_{x} \, d_{x}+d_{x} \, \tau) \big)$ bit operations in the worst case, where $\tau$ is the maximal coefficient bitsize of $P$ and $Q$, and $d_v =\max\big( \deg_v(P),\deg_v(Q) \big)$ with $v \in \{x, y\}$.  

Alternatively, it is possible to localize the maximal $y$-projection of the real solutions of  $\Sigma$ by only applying a linear separating form on the system $\Sigma$, i.e., without computing isolating boxes for the whole system solutions. The linear separating form $t= y\, +\, s \, x$ proposed in \cite{cheng2009root} preserves the order of the solutions of the sheared system $$\Sigma_s=V_{\mathbb R}\left(\left\langle P(x, t- s\, x ), \, Q(x, t-s\, x ) \right\rangle\right)$$ with respect to the $y$-projection of the solutions of the original system $\Sigma$. In other words, with this linear separating form $t= y\, +\, s \, x$, we have:  $$t_1=y_1\, +\, s \, x_1 \, <\, t_2=y_2\, +\, s \, x_2 \, \implies \, y_1 \, \leq \, y_2.$$ Thus, the projection of the solutions of $\Sigma_s$ onto the new separating axis $t$ could be done so that we could simply choose the $y$-projection corresponding to the maximal $t$-projection of the real solutions of $\Sigma_s$. The drawback of this method is in the growth of the size of the coefficients of $\Sigma_s$ 
due to the large size of $s$. The complexity analysis shows that the last algorithm performs, in the worst case,  $\Tilde{{O}}_B\big( d_x^3 \, d_y^4 \, \tau\, \big( d_x^2 + d_x\, d_y + d_y^2  \big) \big)$ bit operations. 

\emph{Real roots counting}. This approach consists in univariate solving with a polynomial of algebraic coefficients. The third method localized the maximal $y$-projection of the system real solutions  $-$ denoted by $\Bar{y}$ $-$ by first isolating the real roots of the univariate {resultant} polynomial ${\rm Res}(P, \frac{\partial P}{\partial x}, x)$ and
then verifying the existence of at least one real root of the greatest common divisor $\gcd \big(P(x, \Bar{y}), \frac{\partial P}{\partial y}(x, \Bar{y})\big) \in \mathbb{R}[x]$ of $P(x, \Bar{y})$ and $\frac{\partial P}{\partial x}(x, \Bar{y})$. But the polynomial $P$  $-$ corresponding to our modeled problem $-$ defines a plane real algebraic curve bounded in the $y$-direction by the value that we are aiming at computing. Thus, a resulting key point is that the number of real roots of $\gcd \big(P(x, \Bar{y}), \frac{\partial P}{\partial x}(x, \Bar{y})\big)$ is equal to the number of real roots of $P(x, \Bar{y})$. Hence, we could simply compute the \emph{Sturm-Habicht sequence} \cite{gonzalez1998sturm} of  $P(x,\Bar{y})$ for counting the number of its real roots without any consequent overhead. Since the $\gcd$ polynomial has a larger size than the polynomial $P$, this key point leads to a better complexity in the worst case. We also mention that the {Sturm-Habicht sequence} corresponding to $P(x, \Bar{y}) \in \R[x]$ is a \emph{signed subresultant sequence} of the polynomials $P(x, \Bar{y}) \in \R[x]$ and its derivative with respect to $x$. Being already computed to obtain ${\rm Res}(P, \frac{\partial P}{\partial x}, x)$, the practical and theoretical complexities are mainly carried by the complexity of evaluating the leading coefficients with respect to $x$ of the subresultant polynomials 
(the so-called \emph{the principal subresultant coefficients}) over the real value $\Bar{y}$. Additionally, if the real plane algebraic curve $P(x, y) = 0$ has no real \emph{isolated singular points}, then the complexity can be further improved since, in this case, the evaluation is done over a rational number instead of an algebraic number. It is worthwhile mentioning that this improvement is theoretically slight. In fact, for evaluating over an algebraic number, say $\Bar{y}$, we are technically evaluating over two rational numbers,  mainly the endpoints of the isolating interval of the algebraic value $\Bar{y}$, where $\Bar{y}$ is a real root of  ${\rm Res}(P, \frac{\partial P}{\partial x}, x)$.    
The complexity analysis showed that this algorithm performs $\Tilde{{O}}_B\big( d_{x}^4 \, d_{y}^2 \, (d_y+\tau) \big)$ bit operations in the worst case and $\Tilde{{O}}_B(d_x^4 \,  d_y^2 \, {\tau})$ when the plane curve $P(x,y)=0$ has no real isolated singular points. 

After implementing the three approaches in {\tt Maple}, the real roots counting approach was shown to perform better than its counterpart.
\section{Parametric case}
\subsection{Background}

Let us consider the \emph{basic semi-algebraic set} defined by 
$$\mS = \{x \in \R^n \; | \;  p_1(x) = 0,\ldots, p_s(x) = 0, \, f_1(x) > 0, \ldots, f_s(x) > 0\}$$
where $p_i$, $f_j$ are polynomials with rational coefficients for $i=1, \ldots, s$. Moreover, let $[U, X] := [u_1,\ldots, u_d, x_{d+1},\ldots, x_n]$ be the set of unknowns or variables, while $U = [u_1,\ldots, u_d]$ is the set of
parameters and $X = [x_{d+1},\ldots, x_n]$ the set of the indeterminate. We denote by $\Pi_U: \C^n \longrightarrow \C^d$ the canonical projection onto the parameter space $(u_1, \ldots, u_d, x_{d+1}, \ldots, x_n) \longmapsto (u_1, \ldots , u_d)$. 
Finally, for any set $\mathcal{V} \subset \C^n$, we will denote by $\overline{\mathcal{V}}$ the \emph{$\C$-Zariski closure} of $\mathcal{V}$, namely, the smallest affine algebraic set containing $\mathcal{V}$.

In this section, we only consider systems that are so-called \emph{well-behaved systems}. They are systems that contain as many equations as indeterminates, are generically zero-dimensional, i.e., for almost all complex parameter values, at most finitely many complex solutions exist, and \emph{generically radical}, i.e., for almost all complex parameter values, there are no solutions of multiplicity greater than $1$ (in particular, the input equations are square-free).
 
In applications, questions that often arise concern the structure of the solution space in terms of the parameters such as, e.g., {determining the parameter values for which real solutions exist or, more generally, determining the parameter values for which the system has a given number of real solutions.}

To solve the well-behaved system $\mathcal S$, it is significant to choose a finite number of
representative ``good'' parameter values that cover all possible cases.

A method was proposed by D.~Lazard and F.~Rouillier in \cite{lazard2007solving} based on the concept of \emph{discriminant variety}. This method can be outlined as follows. First, the set of solutions of the system,  considered as equations in both the parameters and the indeterminates, is projected onto the parameter space. Then, the topological closure $\Bar{S}$ of the resulting projection, which will usually be equal to the whole parameter space $\R^d$, is divided into two parts: a discriminant variety $W$ and its complement $\Bar{S}\backslash W$. A discriminant variety $W$ of the system can be understood as the set of ``bad'' parameter values leading to non-generic solutions of the system, for instance, infinitely many solutions, solutions at infinity, or solutions of multiplicity greater than $1$. It is a generalization of the well-known discriminant of a univariate polynomial. 

The complement of $W$, $\Bar{S}\backslash W$, can be expressed as a finite disjoint union of connected open sets, usually called cells, such that 
a specific behavior of the system does not change when the parameters vary within the same cell. 

Describing the connected open cells can be done using the so-called \emph{Cylindrical Algebraic Decomposition} or CAD for short \cite{dolzmann1999real}. It is a well-known concept in computational real algebraic geometry, first proposed by G.~E.~Collins \cite{collins1975quantifier}. It allows us to obtain a practical description of the connected components of the complement of the discriminant variety.
For a well-behaved system, the discriminant variety is of dimension less than $d$, and it characterizes the boundaries between these cells. Hence, the number of solutions of the system only changes on the boundary or when crossing a boundary. We are interested in open cells of maximal dimension (complement of the discriminant variety), which is
the efficient processing of computations within the cylindrical algebraic decomposition.

For each open cell in the parameter space, we can choose a sample point, evaluate the original system at the sample point, and then solve the resulting non-parametric system. In this way, we can, for instance, determine the number of real solutions of the system that are constant for parameter values chosen in the same open cell.

\begin{example}
Consider the following semi-algebraic system 
$$\mathcal{C} = \big\{a\, x^2 + b -1 =0,\, b\, z +y = 0,\, c\, z +\, y =0, \, c\, >\, 0 \big\}, $$
where $\{a,\, b,\, c\}$ is the set of parameters and $\{x,\, y,\, z\}$ is the set of indeterminates. Then, a discriminant variety is defined by: $$\mathcal{D}=\big\{ (a, \, b \, , c) \; |\; a = 0 \text{ or } b=1 \text{ or } b = c \text{ or } c = 0 \big\}.$$ In fact, the case $a = 0$ corresponds to a vanishing leading coefficient of the first equation which can be interpreted as ``solution at $\infty$''. If $b=1$, then the first equation has a real root $x=0$ of multiplicity $2$. If $b = c$, the second and third equations of $\mathcal{C}$ coincide, and therefore the system $\mathcal{C}$ becomes underdetermined and has infinitely many solutions. Finally, the case $c = 0$ corresponds to a boundary case for the inequality $ c > 0$. Thus, for every choice of the parameter values outside the discriminant variety $\mathcal{D}$, the system 
$\mathcal{C}$ has finitely many solutions, all of multiplicity $1$.
\end{example}

A {\tt Maple} package for solving parametric polynomial systems was introduced in \cite{gerhard2010package} under the name {\tt Parametric}. It is mainly based on techniques such as Gröbner bases, polynomial real root finding, and cylindrical algebraic decomposition. In particular, the command {\tt DiscriminantVariety} of this package computes a discriminant variety of a polynomial system depending on parameters and verifying the required conditions.

\begin{example}
We consider the univariate polynomial $P= x^2 +a \, x +  b$ whose coefficients depend on the parameters set {$\{a, b\}$}. A discriminant variety is thus the curve $a^2 - 4 \, b = 0$. Using CAD, the parameter space can then be decomposed into $4$ cells (see Figure~\ref{fig:regions_plot}) in which $P$ has a constant number of solutions that can be computed using a sample point from each cell.
\begin{figure}[H]\centering\includegraphics[width=0.3\textwidth]{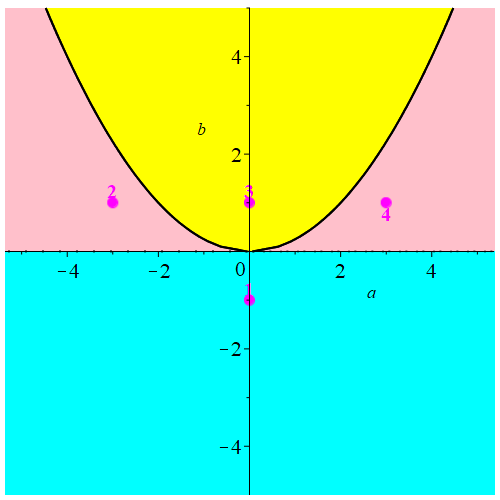}
\caption{Plot of the 4 regions of the parameter space for which the univariate equation $P= x^2 +a \, x +  b $ has exactly two solutions in the pink (regions $2$ and $4$) and cyan (region $1$) cells below the parabola, and no solution in the yellow cell (region $3$) above the parabola.}
\label{fig:regions_plot}
\end{figure}
\end{example}
\subsection{Proposed algorithm}
{

Let $n(\omega, \gamma)$ be defined as in (\ref{def:Phi}) and (\ref{eq:Curve_C}). Note that $n$ belongs to $\mathbb{Z}[\alpha][\omega, \gamma]$, where $\alpha=\alpha_1, \ldots, \alpha_d$ represents a set of real parameters. By considering the following two sets $\Sigma$ and $\Sigma_{\infty}$ 
\[
\left \{ \begin{array}{ll} 
\Sigma = \left \{( \omega, \,  \gamma, \alpha) \in \mathbb{R}^{2+d} \; | \; n(\omega, \gamma, \alpha)=0, \dfrac{\partial n}{\partial \omega}(\omega, \gamma, \alpha)=0 \right\}, \vspace{1mm}
\\ \Sigma_{\infty} = \big\{ (\gamma, \alpha ) \in {\R}^{1 + d} \; | \;  \Lc_{\omega}(n) = 0\big\},
\end{array}
\right.
\] 
then, for $\alpha \in \mathbb{R}^d$, the $L^{\infty}$-norm of the transfer matrix $G$ is given by 
$$\|G \|_{\infty}=\max \left(\pi_\gamma(S), \pi_\gamma(S_{\infty})\right).$$
 
This implies that the norm we are aiming to compute (namely, the maximal $\gamma$ projection of the real solutions of $\Sigma$ and $\Sigma_{\infty}$) is a function of the parameters, and thus, is influenced by the parameters $\alpha$.

Before introducing the proposed algorithm for representing this norm in terms of the parameters, consider the transfer function derived from an application to a gyro-stabilized sight model, as discussed in \cite{rance2018commande}. Notably, the transfer function in this model relies on a specific set of parameters. In Proposition~\ref{norm_with_parameters_proposition},
we will compute its $L^{\infty}$-norm as a function of these parameters through a step-by-step proof. Employing computer algebra tools, including resultant polynomials and intricate properties related to roots' multiplicity, we will demonstrate the process.

\begin{proposition}\label{norm_with_parameters_proposition}
Let $\omega_0,  \ \omega_1 \in \mathbb{R}_{>0}$, $\omega_0 \neq \omega_1$, $0<\xi\leq 1$ and:
\begin{equation}\label{def:G}
G= \dfrac{\left(\dfrac{s}{\omega_0}\right )^2+2 \, \xi \, \left(\dfrac{s}{\omega_0} \right)+1}{\left(\dfrac{s}{\omega_1} \right)^2+2 \, \xi \, \left(\dfrac{s}{\omega_1} \right)+1}.
\end{equation}
Set $r=\omega_1/\omega_0$, $\mu= 4 \, \xi^2 \, (\xi - 1) \, (\xi + 1)$, and let $\delta$ be the maximal real root of:
$$M=\mu \,  \gamma^4+\big((r^2-1)^2 - 2 \, \mu \, r^2\big) \, \gamma^2+\mu \, r^4 \in \mathbb{R}[\gamma].$$
Then, the $L^{\infty}$-norm of $G$ is given by:
\begin{equation}\label{equation parametre}
  \parallel G \parallel_{\infty} \, = \, 
 \left \{
    \begin{array}{ll}
        \max\{1, r^2\} & \; \mbox{if} \; \; \xi \geq \frac{1}{\sqrt{2}},  \vspace{1mm} \\
       \delta & \; \mbox{if} \; \;  \xi < \frac{1}{\sqrt{2}}.
    \end{array}
\right.  
\end{equation}
\end{proposition}

To establish Proposition~\ref{norm_with_parameters_proposition}, we introduce two essential lemmas that will facilitate our proof.

\begin{lemma}\label{lemma_for_N}
Let us consider $\xi, \, \omega_0, \, \omega_1 \in \mathbb{R}_{>0}$, $\omega_0 \neq \omega_1$, $0 < \xi <  1$, $r=\omega_1/\omega_0$, and $\mu= 4 \,  \xi^2 \, (\xi - 1) \, (\xi + 1)$. Then,  the following polynomial 
 $$M_1=\mu \, X^2+\big((r^2-1)^2 - 2 \, \mu \, r^2\big) \, X+\mu \, r^4  \in \mathbb{R}[X]$$ 
 has two positive real roots $X_1$ and $X_2$ verifying $0<X_1<1<X_2$ and {$X_2>r^4$}.
\end{lemma}
\begin{proof}
A discriminant of $M_1$ is $\Delta=(r+1)^2 \, (r-1)^2 \, \big((r^2-1)^2-4 \, \mu \, r^2\big)$. Since $0 < \xi \leq 1$, we get $-\mu \leq 0$, and thus, $\Delta >0$, which shows that $M_1$ has two distinct real solutions, denoted by $X_1$ and $X_2$ with the assumption that $X_1 < X_2$. Moreover, we have
$$X_1 \, X_2=r^4 >0, \quad X_1+X_2=\dfrac{(r^2-1)^2 - 2 \, \mu \, r^2}{-\mu} >0,$$
which yields $X_1 {>0}$ and $X_2=r^4/X_1 >0$. Now, if we let $x:=X-1$, then we get $M_1(X)=M_1(x+1)=m(x)$, where:
$$m(x)=\mu \, x^2+\big((r^2-1)^2-8 \, \xi^2 \, (1-\xi^2) \, (r^2+1)\big) \, x +(2 \, \xi^2-1)^2 \, (r-1)^2 \, (r+1)^2.$$
Clearly, the two roots of $m$ are $X_1-1$ and $X_2-1$, and we have 
$$(X_1-1) \, (X_2-1)=\dfrac{(2 \, \xi^2-1)^2 \, (r-1)^2 \, (r+1)^2}{\mu} <0,$$
which shows that $X_1 < 1$ and $X_2 > 1$ since $X_1 < X_2$. Finally, $X_1 < 1$ yields $X_2=r^4/X_1 > r^4$, which proves the result. 

\end{proof}

\begin{lemma}\label{lemma_for_L}
Let $\omega_0, \,  \omega_1 \in \mathbb{R}_{>0}$, $\omega_0 \neq \omega_1$, $0<\xi< 1$, $\xi \neq 1/\sqrt{2}$, $r=\omega_1/\omega_0$ and $\beta = 2 \, \xi^2 -1$. Then, the following polynomial 
 $$L=\beta \, Y^2+\omega_0^2 \, (r^2+1) \, Y+\beta \, r^2 \, \omega_0^4  \in \mathbb{R}[Y]$$
has two positive real roots if and only if $0<\xi< 1/\sqrt{2}$. \end{lemma}

\begin{proof}
A discriminant of $L$ is $\delta = w_0^4 \, (r^2 +2 \, \beta \,  r + 1)\, (r^2 -2 \, \beta \, r + 1)$. We have $\beta^2-1=4 \, \xi \, (\xi-1)<0$, and thus, the discriminant $4 \, (\beta^2-1)$ of the two polynomials $r^2 +2 \, \beta \,  r + 1$ and $r^2 -2 \, \beta \, r + 1$ is negative, which yields $\delta >0$ and thus, $L$ has two distinct real roots. The product of these roots is $r^2 \, \omega_0^4 >0$ and their sum is $\omega_0^2 \, (r^2+1)/(-\beta)$.
Since $\beta < 0$ if and only if $0 < \xi < 1/\sqrt{2}$, we obtain that the sum is positive if and only if $0 < \xi < 1/\sqrt{2}$, which then implies that $L$ has two positive real roots only when $0 < \xi < 1/\sqrt{2}$. 
\end{proof}

We can now proceed with the proof of Proposition~\ref{norm_with_parameters_proposition} by leveraging  Lemmas~\ref{lemma_for_N} and \ref{lemma_for_L}.

\begin{propproof}[of Proposition~\ref{norm_with_parameters_proposition}]
Set $\alpha = (r, \, \omega_0, \, \xi) \in \R^3$. Let $N$ and $D$ be two polynomials such that:
$$G(-i \, \omega) \, G(i \, \omega)=\dfrac{N(\omega)}{D(\omega)}.$$
Let $n(\gamma, \omega)=D(\omega) \, \gamma^2-N(\omega)$ and consider:
\[
\left \{ \begin{array}{ll}
\Sigma = \left \{(\omega, \, \gamma) \in \mathbb{R}^{2} \; | \; \ n(\omega, \, \gamma)=0,\ \dfrac{\partial n }{\partial \omega}(\omega, \, \gamma)=0 \right\}, \vspace{1mm} \\ \Sigma_{\infty} = \big\{ \gamma \in \R  \; | \;  \Lc_{\omega}(n) = 0\big\}.
\end{array}
\right.
\]

Doing the computation, we obtain $$n=(\gamma^2-r^4) \, \omega^4+2 \, r^2 \, \omega_0^2 \, \beta \, (\gamma^2-r^2) \, \omega^2+r^4 \, \omega_0^4 \, (\gamma^2-1),$$
where $\beta = 2 \, \xi^2 -1$.
The resultant $R$ of $n(\omega, \gamma)$ and $\frac{\partial n}{\partial \omega}(\omega, \gamma)$ with respect to $\omega$ is then defined by
 $$R=256 \, w_0^{12} \, r^{12} \, (\gamma^2-1) \, (\gamma^2-r^4)^2 \, M^2,$$
where $M=\mu \, \gamma^4+((r^2-1)^2 - 2 \, \mu \, r^2) \, \gamma^2+\mu \, r^4$, i.e., $M(\gamma)=M_1(\gamma^2)$ with $M_1$ defined in Lemma~\ref{lemma_for_N}. By assumptions on the parameters, $M_1(X)$ has two positive real solutions, $X_1$ and $X_2$, and thus, $M$ has the four real roots $\pm \sqrt{X_1}$, $\pm \sqrt{X_2}$.  
Thus, based on the properties of the resultant polynomial, we have:  $$\parallel G \parallel _ {\infty} = \max \left\{  \left \{1, \, r^2, \,  \sqrt{X_1}, \, \sqrt{X_2} \right \} \cup \bigg( \pi_{\gamma}(\Sigma) \cup \Sigma_{\infty} \bigg) \right\}.$$
\begin{enumerate}
\item For $\gamma=r^2$: ${\rm Lc}_{\omega}(n)=(\gamma^2-r^4)=0$, i.e, $r^2 \in \Sigma_{\infty}$.
\item For $\gamma=1$, we get:
\[
 \left \{
    \begin{array}{ll}
    n(\omega,\, 1)=(r^2-1) \, \omega^2 \, f_1, \\
      q(\omega,\, 1)=4 \, (r^2-1) \, \omega  \, f_2,
    \end{array}
\right., \quad 
\left \{
    \begin{array}{ll}
    f_1:= (r^2 + 1) \, \omega^2 + 2 \, r^2 \, \omega_0^2 \, (2 \, \xi^2 - 1), \\
    f_2:=(r^2 + 1) \, \omega^2 + r^2 \, \omega_0^2 \, (2 \, \xi^2 - 1).
    \end{array}
\right.
\]

Hence, ${\rm Res}(f_1,\, f_2,\, \omega)=\big(\omega_0 \, r^2 \, (\xi^2+1) \, (r^2+1)\big)^2$ does not vanish. Thus, $\gcd(f_1,f_2)=1$, which yields $\gcd\big( n( \omega, \, 1),\, q(\omega, \, 1)\big)=(r^2-1) \, \omega$ and proves: 
$$(\omega, \, \gamma)=( 0, \, 1) \in \Sigma.$$
\item For $\gamma$ real root of $M$:
the point is to verify that $\gamma \in \pi_{\gamma} (\Sigma)$. 
For doing so, we start by computing $F={\rm Res}(n(\omega, \gamma),\frac{\partial n}{\partial \omega}(\omega, \gamma), \gamma)$. Using the properties of resultants, we recall that $\pi_{\omega}(\Sigma) \subset V_{\R}(F)$, where $F \in \mathbb{R}[\omega]$:

\[
F=c\; {\omega}^2\, F_1^2, \quad 
\left \{
    \begin{array}{ll}
     F_1= \beta\, \omega^4+\omega_0^2\, (r^2+1)\, \omega^2 + r^2\omega_0^4\, \beta, \\
    c= 16\; \omega_0^{20}\; r^8\, (r^2-1)^2.
    \end{array}
\right.
\]

We note that the roots of the polynomial $F_1 \, \in \, \mathbb{R}[\omega]$, are the $\omega$-coordinates of $( \omega_{i,j}, \, \gamma_i) \in V_{\mathbb{R}\times\mathbb{C}}(\langle n,\, \frac{\partial n}{\partial \omega} \rangle)$, where $\gamma_i \in V_{\mathbb{R}}(M)$. We have seen that: $$\pi_{\omega}\bigg(V_{\R} \big(\, \langle n(\omega, \, \pm 1), \frac{\partial n}{\partial \omega}( \omega, \, \pm 1)\,  \rangle \big)\bigg) = 0.$$ Taking into consideration the power and degree of the factor $\omega$ in $F$ and the power and degree of the factor $\gamma^2-1$ in $R$, and the properties of the resultant concerning root multiplicity (see \cite[Chapter 4]{cox1998d}), we can say that: $$(\omega,  \gamma )\in \Sigma, \ \gamma=\pm1 \iff \omega = 0.$$ 
With this being said, let $L\in \mathbb{R}[Y]$ be the polynomial obtained after substituting $\omega^2$ by $Y$ in $F_1$. Based on Lemma~\ref{lemma_for_L}, $L$ has two positive real roots if and only if $\xi<\frac{\sqrt{2}}{2}$. Thus we can conclude two cases:  
\begin{itemize}
    \item For $\xi< \frac{\sqrt{2}}{2}$, $F_1 \in \mathbb{R}[\omega]$ has four real roots. Thus, we have: $$\forall \; \gamma_i \in V_{\mathbb{R}}(M),\ \exists \  \omega_{i,j}\in V_{\mathbb{R}}(F_1): \; (\omega_{i,j}, \gamma_i)\in \Sigma.$$ Consequently, based on Lemma~\ref{lemma_for_N} where we proved that $$X_1<1<X_2, \quad X_2> r^4,$$ we can say that $\delta= \sqrt{X_2}>r^2$, and we  conclude that  $\parallel G \parallel _{\infty}= \delta$.
    \item For $\xi> \frac{\sqrt{2}}{2}$, $F_1$ has no real roots in $\omega$. In this case, none of the real roots of $M$ is a good candidate, and we conclude that: $$\parallel G \parallel _{\infty} \, = \,  \max \, \{1, r^2\}.$$ 
\end{itemize}

\item For $\xi= \frac{\sqrt{2}}{2}$,
$M= c\, (\gamma^2-1) \, (\gamma^2-r^4)$, where $c \in \R$. In this case, we have:
$$\parallel G \parallel _{\infty} \, = \, \max \, \{1,r^2\}.$$

\item Similarly, for $\xi= 1$, $M=c\, \gamma \, (\gamma^2-1) \, (\gamma^2-r^4)$, where $c \in \mathbb{R}$. Then, we have: 
$$\parallel G \parallel _{\infty} \, = \, \max \, \{1,r^2\}.$$  
\end{enumerate}
\end{propproof}

It is important to emphasize that deriving this representation manually, especially in the presence of parameters, can be a daunting task due to its inherent complexities. Therefore, we will proceed to introduce an algorithm specifically designed to compute the $L^{\infty}$-norm in the parametric case. After stating this proposed algorithm, we will apply it once more to the aforementioned example, verifying its efficacy by obtaining the same result.

To establish the context, one method for achieving this parameter space decomposition and norm representation is by employing a CAD of $\R^{d+2}$ adapted to $\{n(\omega, \, \gamma, \alpha) = 0,\, \frac{\partial n}{\partial \omega}(\omega, \, \gamma, \alpha) =0 \}$ and a semi-algebraic set $S_p$ containing the inequalities verified by the parameters. But it may yield a very huge result, difficult to analyze in practice, with lots of cells we are not interested in. We are just interested in the cells where the curves representing the $\gamma$-projection of the system solutions, when they exist, are continuous and do not intersect.

Below, we state Algorithm~\ref{algoP}, where {\tt non\_parametric} corresponds to one of the proposed methods in Section~\ref{Problem description: non parametric case} and $\tt index_i$ corresponds to the index of the output of {\tt non\_parametric} in the sorted set of the real roots of $R = {\rm Res}(n, \frac{\partial n}{\partial \omega}, \omega)$.

This algorithm uses a CAD to study 
$R = {\rm Res}(n, \frac{\partial n}{\partial \omega}, \omega)$. This resultant polynomial depends on the set of parameters $\alpha$. Within this context, the discriminant variety simplifies to the discriminant of this resultant polynomial. The cells derived in the parameter space, post CAD application, represent regions where the roots $\gamma$ of the resultant exhibit non-variant behavior. Notably, within each cell, these roots do not intersect, ensuring their distinct positions relative to one another. By selecting a sample point within a cell, we can determine the ``maximal" $\gamma$, which is a parameterized expression depending on $\alpha$.

\begin{algorithm}
  \caption{Parametric case}
  \label{algoP}
    \textbf{Input:} A well-behaved polynomial system $\Sigma = \big\{n(\omega, \, \gamma) = 0, \frac{\partial n }{\partial \omega}(\omega, \, \gamma)=0 \big\}$, $n\in  \Z[\alpha_1,\ldots,\alpha_d ][\omega, \, \gamma]$, and a semi-algebraic set $S_p$ (for parameters conditions).\\
    \textbf{Output:} A list of pairs [$C_i$, {\tt index\_i}] as in Theorem~\ref{theorem:1}.
\begin{enumerate}
   \item Compute $R = {\rm Res}(n, \frac{\partial n}{\partial \omega}, \omega)$.
    \item Compute $\mathrm{R}_2$, a discriminant (variety) of $\{R=0, S_p\}$ with respect to $\pi_{\alpha}$.
    \item Using a {\tt CAD}, compute the partition $\{C_1,\dots, C_l\}$ of $C = \mathbb{R}^d \backslash \mathrm{R}_2 $, along with sample points ${\tt sample_i} \in C_i$.
    \item Apply {\tt non\_parametric} on ${\tt subs(sample_i , \Sigma)}$ and get {\tt index\_i}, 
    \item {\tt return} \big\{[$C_i$, {\tt index\_i}], $i= 1, \ldots, l$ \big\}.
\end{enumerate}
\end{algorithm}

\begin{theorem}\label{theorem:1}
Given a well-behaved system $\Sigma = \big\{n(\omega, \, \gamma) = 0, \frac{\partial n }{\partial \omega}(\omega, \, \gamma) = 0  \big\}$, where $n\in \Z[\alpha][\omega, \, \gamma]$ and $\alpha = (\alpha_1, \ldots, \alpha_d)$, and a semi-algebraic set $S_p$ verified by a set of parameters $\alpha$, Algorithm~\ref{algoP} outputs a set of pairs $[C_i , {\tt index_i}]$, where the $C_i$s are disjoint parameter cells in the parameter space and ${\tt index_i}$ denotes the the index of $$\gamma_{\max} = \max \left \{ \pi_{\gamma} \left ( V_{\mathbb{R}^2} \left(
\left \langle 
{n}, \frac{\partial {n}}{\partial \omega} 
\right \rangle 
\right)\right) \cup V_{\mathbb{R}}\left(\langle {\rm Lc}_{\omega}({n}) \rangle \right) \right\}$$ in the sorted set of the real roots of $R = {\rm Res}(n, \frac{\partial n}{\partial \omega}, \omega)$ over a cell $C_i$. 
\end{theorem}
\begin{proof}
In the first step, we compute a discriminant variety $R$ of $\big\{ n(\omega, \, \gamma, \alpha) = 0\big\}$ with respect to $\pi_{(\alpha,  \, \gamma)}$, where $n(\omega, \, \gamma, \alpha)$ is seen as a univariate polynomial in $\omega$. We recall that this discriminant variety is the set of parameter values leading to non-generic solutions of the system, for example, infinitely many solutions, solutions at infinity, or solutions of multiplicity greater than $1$. It is simply the resultant polynomial of the polynomial $n(\omega, \, \gamma)$ and its derivative with respect to the main variable $\omega$, i.e., $R = \Res\big({n}(\omega ,\gamma), \frac{\partial {n}}{\partial \omega}(\omega, \gamma), \omega \big) \in \Q[\alpha][\gamma]$. In this case, $R = 0$ is a sub-variety of $\R^{d+1}$ and the complement of $R = 0$, $\R^{d+1}\backslash \{R = 0\}$, can be expressed as a finite disjoint union of cells, which are connected open sets, such that the number of the system real solutions $\omega$ does not change when the parameters vary within the same cell.

In the second step, we consider the variable $\gamma$ as the main variable in the polynomial $R \in \Q[\alpha,\gamma]$ since it is the polynomial embodying the $\gamma$-projection of the system $\big\{ n(\omega, \, \gamma) = 0,\ \frac{\partial n}{\partial \omega}(\omega, \, \gamma) = 0\big\}$. In this case, the $\gamma$-projection of the system solutions is considered as a real function of $\alpha$ such that the position of the curves representing $\gamma(\alpha) = 0$ changes after each intersection of at least two curves in $\R^d$. To locate the maximal value $\gamma$ over a given cell, we decompose $\R^d$ into cells where no changes in the position of the curves $\gamma(\alpha)$ occur. For doing so, we can naturally propose to  eliminate from the parameter space the set of ``bad'' parameter values leading to non-generic solutions of 
$R = 0$. This set corresponds to a discriminant variety of $ \{(\alpha, \gamma) \in \mathbb{R}^{d+1} \; | \;  R=0\} \, \cup \, S_p$ with respect to $\pi_{\alpha}$, denoted by $R_2$. We recall that $R_2$ is simply the curve of the discriminant of $R$ with respect to the variable $\gamma$, multiplied by the leading coefficient of $R$ with respect to $\gamma$, that is nothing but ${\rm Res}\left(R, \frac{\partial R}{\partial \gamma}, \gamma \right) \in \Z[\alpha]$, up to some curves related to the inequalities of $S_p$.

In the third step, using a CAD, we can decompose $ C =  \mathbb{R}^d \backslash R_2$ into connected cells, above each cell, the variable $\gamma$ is a real-valued function depending continuously on the parameters $\alpha$, whose graphs are disjoint. Now, 
let $\big\{{C}_1,\dots, {C}_l\big\}$ be the partition of ${C}$ and let ${\tt sample_i}$ a sample point in $C_i$.

By substituting $\alpha$ by ${\tt sample_i}$ in $\Sigma$, we are now left with a zero-dimensional polynomial system $\Sigma_i$ in $\mathbb Z[\omega, \gamma]$. Thus, applying {\tt non\_parametric} in the fourth step, we can isolate $\gamma_{\max}$ 
as the maximal real root of $R = {\rm Res}(n, \frac{\partial n}{\partial \omega}, \omega)$ corresponding to a real point $(\omega,\gamma)$. 

Now, consider $\tt index_i$ to be the index of $\gamma_{\max}$ in the sorted set of the real roots of $R$. We know that over a cell $C_i$, the real roots $\gamma$ of $R$ are represented as real-valued functions depending continuously on the parameters, whose graphs are disjoint. This assures that for every parameter value in $C_i$, the relative positions of the functions $\gamma$ remain the same. Therefore, for every ${\tt sample_i}$ in $C_i$, $\tt index_i$ is the same. Thus, Algorithm~\ref{algoP} outputs a set of pairs $[C_i , {\tt index_i}]$, where the $C_i$'s are disjoint parameter cells and ${\tt index_i}$ denotes the index of $\gamma_{\max}$ in the sorted set of the real roots of $R = {\rm Res}(n, \frac{\partial n}{\partial \omega}, \omega)$ over a cell $C_i$. 
\end{proof}

\begin{example}\label{Example parametric case 2}
We consider again the transfer function $G$ studied in Proposition~\ref{norm_with_parameters_proposition}, i.e., (\ref{def:G}), 
for  $S_p = \{\omega_0 >0,\, \omega_1>0,\, \omega_0 \neq \omega_1,\, 0<\xi\leq 1\}$. We can apply Algorithm~\ref{algoP} to the polynomial system $\Sigma = \big\{n(\omega, \, \gamma) =0, \frac{\partial n }{\partial \omega}(\omega, \, \gamma) =0 \big\}$ and $S_p$ in order to represent the ${L}^{\infty}$-norm of $G$.

Let $R={\rm Res}(n, \frac{\partial n }{\partial \omega}, \omega) \in \mathbb{Z}[\mathbf{\alpha}, \gamma],$ where $n$ and $R$ defined in the proof of Proposition~\ref{norm_with_parameters_proposition}. Doing the computation using the {\tt Maple} functions {\tt DiscriminantVariety} and {\tt CellDecomposition} of the {\tt RootFinding[Parametric]} package, a discriminant variety $R_2$ is given by the union of the curves defined by the following polynomials:$$\bigg\{r,\, \omega_0,\,  \xi,\,  r - 1,\,  r + 1,\,  \xi - 1,\,  \xi + 1,\,  2\, \xi^2 - 1,\,  -4 \, r\, \xi^2 + r^2 + 2 \, r + 1,\,  4 \, r \, \xi^2 + r^2 - 2 \, r + 1 \bigg\}.$$

Computing a CAD of $C = \mathbb{R}^d \backslash R_2$, we get the partition $\{C_1,\dots, C_4\}$, where
\begin{center}
  \begin{itemize}
    \item[] $C_1=  \bigg\{ 0< \xi < \sqrt{2}/2 \bigg\} \cap \bigg\{ 0< r < 1\bigg\} \cap \bigg\{ \omega_0 >0 \bigg\}$, \vspace{1mm}
    \item[] $C_2= \bigg\{ 0< \xi < \sqrt{2}/2 \bigg\} \cap \bigg\{ r > 1\bigg\} \cap \bigg\{ \omega_0 >0 \bigg\}$, \vspace{1mm}
   \item[]  $C_3= \bigg\{ \sqrt{2}/2< \xi < 1 \bigg\} \cap \bigg\{ 0< r < 1\bigg\} \cap \bigg\{ \omega_0 >0 \bigg\}$, \vspace{1mm}
   \item[] $C_4=\bigg\{ \sqrt{2}/2< \xi < 1 \bigg\} \cap \bigg\{  r > 1\bigg\} \cap \bigg\{ \omega_0 >0 \bigg\}$,
\end{itemize}
\end{center}
and the sample points
\begin{center}

\begin{itemize}
        \item[] ${\tt sample_1}=\left [\xi =\dfrac{25476206690102465}{72057594037927936}, \, r = 1/2, \,  \omega_0 = 1 \right]$,  \vspace{2mm}
        \item[] ${\tt sample_2}=\left [\xi =\dfrac{25476206690102465}{72057594037927936}, \, r = 2, \, \omega_0 = 1 \right]$,  \vspace{2mm}
         \item[]  ${\tt sample_3}=\left[\xi =\dfrac{30752501854533959}{36028797018963968}, \, r = 1/2, \, \omega_0 = 1 \right ]$,  \vspace{2mm}
          \item[] ${\tt sample_4}=\left [\xi =\dfrac{30752501854533959}{36028797018963968}, \, r = 2, \, \omega_0 = 1\right]$.
    \end{itemize}  
\end{center}    

Substituting $\alpha$ by ${\tt sample}_1$ in $n(\omega, \, \gamma)$, we obtain 
$$ n (\omega, \, \gamma) = \left(\gamma - \frac{1}{4} \right) \, \left(\gamma + \frac{1}{4} \right) \, \omega^4 + ( a \, \gamma^2 +b) \, \omega^2 + \frac{1}{16}(\gamma - 1) \, (\gamma +1),$$     
 where:
 \[
\left \{ \begin{array}{ll}
a= -  \dfrac{1947111321950592219128255965533823}{5192296858534827628530496329220096}, \vspace{2mm} \\
b=\dfrac{1947111321950592219128255965533823}{20769187434139310514121985316880384}.
\end{array}
\right.
\]
 
Using one of the proposed methods described in Section~\ref{Problem description: non parametric case}, we obtain the following isolating interval for  $\gamma_{\max}$ $$\bigg[\frac{6100687164736347533}{4611686018427387904}, \frac{24402748658945394263}{18446744073709551616}\bigg],$$ which is the isolating interval of the element of index $8$ in the sorted set of the real roots of $R = {\rm Res}(n, \frac{\partial n }{\partial \omega}, \omega) \in \Z[\gamma]$. Finally, we get $[C_1, {\tt index}_1=8]$ as an element returned in the output list. To match this result with Proposition~\ref{norm_with_parameters_proposition}, we can see that the real roots of $R$ are ordered as follows $$\left \{-\sqrt{X_2},\, -1,\, -r^2,\, -\sqrt{X_1},\, \sqrt{X_1},\, r^2,\, 1,\, \sqrt{X_2} \right\}$$ over $C_1$ and the element of index $8$ is indeed the $L^{\infty}$-norm of $G$ as proven in Proposition~\ref{norm_with_parameters_proposition}. 
 
Similarly, after substituting $\alpha$ by ${\tt sample}_2$ in $n$, 
$\gamma_{\max}$ is of isolating interval $$\bigg[\frac{6100687164736347533}{1152921504606846976}, \frac{ 24402748658945394263}{4611686018427387904}\bigg],$$ which is also the isolating interval of the element of index $8$ in the sorted set of the real roots of $R = {\rm Res}(n, \frac{\partial n }{\partial \omega}, \omega) \in \Z[\gamma]$. Thus, the second element of the output list of Algorithm~\ref{algoP} is $[C_2, {\tt index}_2=8]$. With the notations of Proposition~\ref{norm_with_parameters_proposition}, the real roots of $R$ are ordered as $\{-\sqrt{X_2},\, -r^2,\, -1,\, -\sqrt{X_1},\, \sqrt{X_1},\, 1,\, r^2,\, \sqrt{X_2}\}$ over $C_2$ and we can see that the result matches Proposition~\ref{norm_with_parameters_proposition}.

Following the same approach and substituting $\alpha$ by ${\tt sample}_3$ in $n$, we obtain that $\gamma_{\max}$ is of isolating interval $\big[1, 1\big],$ which is the isolating interval of the element of index $7$ in the sorted set of the real roots of $R$. The third element of the output list of Algorithm~\ref{algoP} is thus  $[C_3, {\tt index}_3=7]$. Moreover, we can verify the result of Proposition~\ref{norm_with_parameters_proposition} where the real roots of $R$ are ordered as $\{-\sqrt{X_2},\, -1,\, -r^2,\, -\sqrt{X_1},\, \sqrt{X_1},\, r^2,\, 1,\, \sqrt{X_2}\}$ over $C_3$. 

Finally, substituting $\alpha$ by ${\tt sample}_4$ in $n$,
we obtain that $\gamma_{\max}$ is of isolating interval $\big[4, 4\big]$ which is the isolating interval of the element of index $7$ in the sorted set of the real roots of $R$. The fourth element of the output list of Algorithm~\ref{algoP} is thus  $[C_4, {\tt index}_4 = 7]$. The real roots of $R$ are ordered over $C_4$ as $\{-\sqrt{X_2},\, -r^2,\, -1,\, -\sqrt{X_1},\, \sqrt{X_1},\, 1,\, r^2,\, \sqrt{X_2}\}$, and we can see that the result also matches Proposition~\ref{norm_with_parameters_proposition}.
\end{example}

\bibliographystyle{splncs04}
\bibliography{smple}


\end{document}